\documentclass[12pt]{article}
\usepackage[T1]{fontenc}
\usepackage[latin9]{inputenc}
\usepackage{float}
\usepackage{units}
\usepackage{amsmath}
\usepackage{amsthm}
\usepackage{graphicx}
\usepackage{amssymb}
\usepackage{indentfirst}
\usepackage{textcomp}
\usepackage{algorithmic}
\usepackage{algorithm}
\newtheorem{theorem}{Theorem}
\newtheorem{lemma}{Lemma}
\newtheorem{definition}{Definition}
\numberwithin{equation}{section}

\begin{document}
  \begin{flushleft}
  {\footnotesize \bf Gakuto International Series}\\
  {\footnotesize \bf Mathematical Sciences and Applications Vol.34 (2011)}\\
  {\footnotesize \it International Symposium on Computational Science 2011}\\
  {\footnotesize pp. 139-150}
  \end{flushleft}
  \vspace{1.5cm}
                   \begin{center}
   {\large \bf Generalized Extended Hamming Codes over\\[0.17cm]
         Galois Ring of Characteristic $2^{n}$}\\[1cm]
                   \end{center}

                   \begin{center}
                   {\sc Muhammad Ilyas$^{1,2}$ and Mieko Yamada$^3$} \\
{\footnotesize $^1$Graduate School of Natural Sciences and Technology, Kanazawa University, Japan,}\\
{\footnotesize $^2$Faculty of Mathematics and Natural Sciences, Bandung Institute of Technology, Indonesia,}\\
{\footnotesize $^3$Institute of Science and Engineering, Kanazawa University, Japan}

   (reiken7@gmail.com)
                    \end{center}
\vspace{2cm}

\begin{abstract}
In this paper, we introduce generalized extended Hamming codes over Galois rings $GR(2^n,m)$ of characteristic $2^n$ with extension degree $m$. Furthermore we prove that the minimum Hamming weight of generalized extended Hamming codes over $GR(2^n,m)$ is 4 and the minimum Lee weight of generalized extended Hamming codes over $GR(8,m)$ is 6 for all $m \geq 3$.
\end{abstract}

        \vspace{2cm}

\vfill
\noindent
--------------------------------------------------------------------
----------
\\
{\footnotesize Received May 25, 2011.\\
This work is supported by JSPS KAKENHI (20540014).\\
AMS Subject Classification 94B05}

%%%%%%%%%%%%%%%%%%%%%%%%%%%%%%%%%
\newpage
\section{Introduction}

Preparata codes are non-linear double-error-correcting codes over finite fields, named after Franco P. Preparata who first described them in 1968. Preparata codes are nonlinear codes and these have more codewords than any comparable linear code presently known. Hammons et al. \cite{hammons} introduced Preparata codes and Kerdock codes over the Galois ring of characteristic 4. The binary image of the Preparata code under the Gray map has the same properties of the Preparata original code over a finite field. They also show that these Preparata code is a dual code of Kerdock code over the Galois ring of characteristic 4.

We define generalized extended Hamming codes over Galois ring $GR(2^n,m)$ of characteristic $2^{n}$ with extension degree $m$ similarly to the definition of the extended Hamming codes over finite fields. If we restrict the characteristic to 4, then the generalized extended Hamming code turns out the Preparata code that Hammons et al. introduced. We show generalized extended Hamming codes have similar properties as the extended Hamming codes over finite fields. By using 2-adic representation, we determine the minimum Hamming weight over $GR(2^n,m)$ and the minimum Lee weight of a generalized extended Hamming code over $GR(8,m)$ by transforming the codewords with those minimum weights of the generalized extended Hamming code over $GR(4,m)$. We could use this method recursively to obtain the minimum Lee weight of generalized extended Hamming codes for general characteristics $2^{n}$.

\section{Galois Rings}

\subsection{Galois Rings $GR(q,m)$}

We let $q=2^{n}$, $n \geq 1$ and $\mathbb{Z}_{q}=\mathbb{Z}/q\mathbb{Z}$ and we put $N=2^m-1$, for $m \geq 3$.

Let $h_{2}\left(X\right)$ be a primitive polynomial over a finite field $\mathbb{F}_{2}$ with $2$ elements. There exist a unique monic irreducible polynomial $h_{q}\left(X\right)\in\mathbb{Z}_{q}\left[X\right]$ of degree $m$ such that $h_2\left(X\right)\equiv h_{q}\left(X\right)$(mod 2) and $h_q\left(X\right)$ divides $X^{N}-1 \in \mathbb{Z}_q\left[X\right]$. We call the polynomial $h_{q}\left(X\right)$ a \textbf{primitive basic irreducible polynomial}.

\begin{definition} [Galois Rings]
Let $\xi$ be a root of $h_{q}\left(X\right)$, so that $\xi^{N}=1$. The \textbf{Galois ring} $GR\left(q,m\right)$ is defined to be $\mathbb{Z}_{q}\left[X\right]/h_{q}\left(X\right)$ and isomorphic to $\mathbb{Z}_{q}\left(\xi\right)$. We put $R=GR\left(q,m\right)$.
\end{definition}

We can write every element $c\in R$ as unique 2-adic representation
\[
c=a_{0}+2a_{1}+4a_{2}+\ldots+2^{n-1}a_{n-1},\]
where $a_{i}\in\left\{ 0,1,\xi,\xi^{2},\ldots,\xi^{N-1}\right\}$, $0 \leq i \leq n-1$.

We know $\xi^j$ can be represented as
\[
\xi^j = a_{0,j} + a_{1,j} \xi + \ldots + a_{m-1,j} \xi^{m-1},\]
where $a_{0,j},a_{1,j},\ldots,a_{m-1,j} \in \mathbb{Z}_q$, $0\leq j\leq N-1$.

Note that $R$ contains zero divisors, these are the elements of the radical $2GR(q,m)=2GR(\frac{q}{2},m)$, the unique maximal ideal in $R$.

\newpage
We denote the set $\{0,1,\xi,\ldots,\xi^{N-1}\}$ by $\mathcal{T}$. We define the homomorphism $\eta$
\begin{equation}
\eta:GR\left(q,m\right)\rightarrow \mathbb{F}_{2^m} \mbox{ as } \eta\left(\xi\right) = \theta \mbox{ and } \eta\left(a\right) = a\pmod{2} \mbox{ for $a \in \mathbb{Z}_q$},\label{eq:etamap}\end{equation}
where $\theta$ is a root of a primitive polynomial $h_2(X)$.

Next, we define the map $\tau:GR(q,m)\rightarrow \mathcal{T}$ as
\begin{equation}
\tau\left(c\right)=c^{2^m}=a_0.\label{eq:taumap}\end{equation}

\subsection{Frobenius automorphism of $GR(2^n,m)$}

In this subsection, we define the Frobenius automorphism over $GR(2^n,m)$. We will use this automorphism in transforming a codeword.

\begin{definition}[Frobenius Automorphism of $GR(2^n,m)$]

The \textbf{Frobenius automorphism} $f$ from $GR(2^n,m)$ to $GR(2^n,m)$ is defined as

\begin{equation}
c^{f}=a_{0}^{2}+2a_{1}^{2}+4a_{2}^{2}+\ldots+2^{n-1}a_{n-1}^2,\label{eq:frobenius}\end{equation}
for any element $c=a_0+2a_1+4a_2+\ldots+2^{n-1}a_{n-1} \in GR(2^n,m)$.
\end{definition}

%%%%%%%%%%%%%%%%%%%%%%%%%%%%%%%%%%%%%%%%%%%%%%%%%%%%%%%%%%%%%%%%%%%%%%%%%%%%%%%%%%%%
\section{Codes over $\mathbb{Z}_q$}

\subsection{Codes over $\mathbb{Z}_q$ and the minimum weight}

We define a code over $\mathbb{Z}_q$ similarly to the definition of a code over a finite field.

\begin{definition}[Codes over $\mathbb{Z}_q$]
If $C$ is an additive subgroup of $\mathbb{Z}^{m}_{q}$, then we call $C$ a \textbf{linear block code} of length $m$ over $\mathbb{Z}_{q}$.
\end{definition}

We need to define the minimum Hamming weight and the minimum Lee weight. Let $c=\left(c_{0},c_{1},\ldots,c_{N-1}\right)\in\mathbb{Z}_{q}^{N}$
be a codeword of $C$.

\begin{definition}[Hamming Weight]
The number of nonzero elements of the codeword $c \in C$ is called the \textbf{Hamming weight} $w_{H}\left(c\right)$ of $c$ and the \textbf{minimum Hamming weight} of $C$ is defined as\[
w_{H}^{m}\left(C\right)=\min_{c\in C,c\neq0}w_{H}\left(c\right).\]
\end{definition}

\begin{definition}[Lee Weight]
The \textbf{Lee weight} $w_{L}\left(c\right)$ of the codeword $c \in C$ is defined as\[
w_{L}\left(c\right)=\sum_{i=0}^{N-1}\min\left\{ c_{i},2^{n}-c_{i}\right\},\]
and the \textbf{minimum Lee weight} of $C$ is defined as \[
w_{L}^{m}\left(C\right)=\min_{c\in C,c\neq0}w_{L}\left(c\right).\]
\end{definition}

\subsection{Generalized Extended Hamming Codes}
Hammons et al. \cite{hammons} worked on generalized extended Hamming codes and called them the "Preparata" codes over $GR(4,m)$, because the binary image of these codes over $GR(4,m)$ under the Gray map has the same properties of Preparata original codes over finite fields. Based on their paper, we define the generalized extended Hamming codes similarly to the definition of the extended Hamming codes over finite fields.

\begin{definition}[Generalized extended Hamming Codes]
We define the parity-check matrix $P$, indexed by $\infty, 0, 1, \ldots, N-1$, as follows :
\[
P = \left(\begin{array}{cccccc}
1 & 1 & 1 & 1 & \ldots & 1\\
0 & 1 & \xi & \xi^{2} & \ldots & \xi^{N-1}\end{array}\right)\]
where $\xi$ is a root of a primitive basic irreducible polynomial. Then $P_q = \{c \in \mathbb{Z}_q^{N+1} | Pc^T = 0\}$ is a code over $GR(q,m)$ and called a generalized extended Hamming code. In what follows, we abbreviate generalized extended Hamming codes as GEH codes.\label{def:P}
\end{definition}

In other words, if $c=\left(c_{\infty},c_{0},\ldots,c_{N-1}\right)$ is a codeword of a GEH code, then it must satisfy the following two equations :
\begin{eqnarray}
c_{\infty}+\sum_{i=0}^{N-1}c_{i} & = & 0,\label{eq:P1}\\
\sum_{i=0}^{N-1}c_{i}\xi^{i-1} & = & 0.\label{eq:P2}\end{eqnarray}

We call $c_{\infty}$ a zero-sum check symbol. Notice that a GEH code is the dual code of a $1^{st}$ order Reed-Muller code over $GR(q,m)$.

If $q = 2$, a GEH code coincides to an extended Hamming code over a finite field $\mathbb{F}_{2^m}$.

\section{Main Theorems}

We can determine the number of codewords, and the minimum Hamming weight of a GEH code $P_q$ and the minimum Lee weight of a GEH code $P_8$.

\begin{theorem}[Number of codewords]
The number of codewords of a GEH code $P_q$ is $q^{k}$, where $k=2^m-m-1$.
\end{theorem}

\begin{proof}

We give the number of codewords by observing a parity-check matrix of a GEH code. For simplicity, we can write the generator matrix as
\begin{equation}
P=\left(\begin{array}{ccc}
1 & \textbf{1} & \textbf{1}\\
\textbf{0} & I_{m} & A_{k}\end{array}\right)\label{eq:generator}\end{equation}
where $I_{m}$ is an identity matrix and $A_{k}$ is $m\times k$-sized square matrix where $k=2^{m}-m-1$ and \textbf{0} is the all-zero column vector and \textbf{1} is the all-one row vector. We can see the first column and $I_{m}$-block act as nullifier. Thus, the number of codewords depends only on $A_{k}$-block, where it is a part of a codeword with length $k$. Therefore, the number of codewords is as assigning $q$ values on $k$ column, that is $q^{k}$ codewords. \end{proof}

\begin{theorem}
The minimum Hamming weight of $P_{q}$ with length $2^{m}$ over $GR(q,m)$ is $4$.\label{thm:ham}
\end{theorem}

\begin{proof}

Hammons et al. \cite{hammons} showed that the minimum Hamming weight of $P_{4}$ is 4. This means that the codeword with the minimum Hamming weight has 4 non-zero entries.

Let $\textbf{c} = ( c_\infty,c_0,\ldots,c_{N-1} )$ be the codeword with Hamming weight 4. Then we see that $2\textbf{c} = ( 2c_\infty,2c_0,\ldots,2c_{N-1} )$ is a codeword of $P_8$ and the Hamming weight of this codeword is 4. For a general case, we can take this method recursively. Thus the minimum Hamming weight of $P_q$ is 4. \end{proof}

\begin{lemma}
Let $q'=\frac{q}{2}$. The minimum Lee weight of $P_{q'}$ over $GR(q',m)$ are always less than or equal to the minimum Lee weight of $P_q$ over $GR(q,m)$. \label{lemma:min}
\end{lemma}

\begin{proof}
We can write the Lee weight of a codeword $c'$ of $P_{q'}$ over $GR(q',m)$ as
\[
w_{L}\left(c'\right)=\sum_{i=1}^{\frac{q'}{2}-1} i\left(n'_i + n'_{q'-i} \right) + \frac{q'}{2} n'_{\frac{q'}{2}}\]
with the same summation index, we can write the Lee weight of a codeword $c$ of $P_q$ over $GR(q,m)$ as
\[
w_{L}\left(c\right)=\sum_{i=1}^{\frac{q'}{2}-1} \{i\left(n_i + n_{q-i} \right) + \left(q'-i\right) \left(n_{q'-i} + n_{q'+i}\right)\} + \frac{q'}{2} \left(n_{\frac{q'}{2}} + n_{\frac{3q'}{2}}\right) +
q' n_{q'}\]
where $n'_i$ and $n_i$ are the numbers of entries $i$ of $c'$ and $c$ respectively.

Define the map $\mu : GR(q,m) \rightarrow GR(q',m)$ as $\mu\left(\xi_q\right) = \xi_{q'}$ and $\mu(c)\equiv c \pmod{q'}$ where $\xi_q$ and $\xi_{q'}$ are roots of the primitive basic irreducible polynomials of $P_q$ and $P_{q'}$ respectively and $c \in GR(q,m)$. Then $\mu\left(P_q\right) = P_{q'}$. From $n'_i = n_i + n_{q'+i}$, we have \begin{eqnarray}
w_{L}\left(c'\right) & = & \sum_{i=1}^{\frac{q'}{2}-1} i\left(n'_i + n'_{q'-i} \right) + \frac{q'}{2} n'_{\frac{q'}{2}} \nonumber \\
& = & \sum_{i=1}^{\frac{q'}{2}-1} \{ i\left(n_i + n_{q-i} \right) + i \left(n_{q'-i} + n_{q'+i}\right)\} + \frac{q'}{2} \left(n_{\frac{q'}{2}} + n_{\frac{3q'}{2}}\right) \nonumber \\
& \leq & \sum_{i=1}^{\frac{q'}{2}-1} \{ i\left(n_i + n_{q-i} \right) + \left(q'-i\right) \left(n_{q'-i} + n_{q'+i}\right)\} + \frac{q'}{2} \left(n_{\frac{q'}{2}} + n_{\frac{3q'}{2}}\right) + q' n_{q'}. \nonumber
\end{eqnarray}

This means that $w_{L}\left(c'\right) \leq w_{L}\left(c\right)$. Thus the minimum Lee weight of $P_{q'}$ over $GR(q',m)$ is always less than or equal to the minimum Lee weight of $P_q$ over $GR(q,m)$.
\end{proof}

Lemma \ref{lemma:min} gives an important information such that the codeword of $P_q$ may have the minimum Lee weight of which image by the map $\mu$ has the minimum Lee weight of $P_{q'}$.

Thus, we show the existence of the codewords of the extended Hamming code with the Hamming weights 4 and 6 over $\mathbb{F}_{2^m}$ concretely. We consider the matrix $P$ as a parity check matrix of an extended Hamming code. Let $\theta$ be a root of $h_2(X)$.

For $m \geq 2$, we know the sum of the following 4 vectors of length $m+1$ is \textbf{0}.
\[
\left(\begin{array}{c}
1\\
0\\
0\\
0\\
\vdots\end{array}\right),\left(\begin{array}{c}
1\\
1\\
0\\
0\\
\vdots\end{array}\right),\left(\begin{array}{c}
1\\
0\\
1\\
0\\
\vdots\end{array}\right),\left(\begin{array}{c}
1\\
1\\
1\\
0\\
\vdots\end{array}\right).\]
The first vector corresponds to the column vector $\left(\begin{array}{c}1\\0\end{array}\right)$ of $P$ in Definition \ref{def:P}, the second and the third vectors correspond to the column vectors $\left(\begin{array}{c}1\\1\end{array}\right)$ and $\left(\begin{array}{c}1\\\theta\end{array}\right)$ of $P$ respectively, and the fourth vector corresponds to the column vector $\left(\begin{array}{c}1\\\theta^i\end{array}\right)$ of $P$ for some $i$. If we put $c_\infty$, $c_0$, $c_1$ and $c_i$ are all 1 and the other entries are zero, then the Hamming weight of $c = \left(c_\infty, c_0, c_1, \ldots, c_{N-1}\right) = \left(1, 1, 1, 0, \ldots, 0, 1, 0, \ldots, 0\right)$ is 4.

For $m \geq 4$, we know the sum of the following 6 vectors of length $m+1$ is \textbf{0}.
\[
\left(\begin{array}{c}
1\\
0\\
0\\
0\\
0\\
\vdots\end{array}\right),\left(\begin{array}{c}
1\\
1\\
0\\
0\\
0\\
\vdots\end{array}\right),\left(\begin{array}{c}
1\\
0\\
1\\
0\\
0\\
\vdots\end{array}\right),\left(\begin{array}{c}
1\\
0\\
0\\
1\\
0\\
\vdots\end{array}\right),\left(\begin{array}{c}
1\\
0\\
0\\
0\\
1\\
\vdots\end{array}\right),\left(\begin{array}{c}
1\\
1\\
1\\
1\\
1\\
\vdots\end{array}\right).\]
The first vector corresponds to the column vector $\left(\begin{array}{c}1\\0\end{array}\right)$ of $P$, the second to the fifth vectors correspond to the column vectors $\left(\begin{array}{c}1\\1\end{array}\right)$, $\left(\begin{array}{c}1\\\theta\end{array}\right)$, $\left(\begin{array}{c}1\\\theta^2\end{array}\right)$, $\left(\begin{array}{c}1\\\theta^3\end{array}\right)$ of $P$ respectively, and the sixth vector corresponds to the column vector $\left(\begin{array}{c}1\\\theta^i\end{array}\right)$ of $P$ for some $i$. If we put $c_\infty$, $c_0$, $c_1$, $c_2$, $c_3$ and $c_i$ are all 1 and the other entries are zero, then the Hamming weight of $c = \left(c_\infty, c_0, c_1, \ldots, c_{N-1}\right) = \left(1, 1, 1, 1, 1, 0, \ldots, 0, 1, 0, \ldots, 0\right)$ is 6.

\vspace{0.5cm}

We give the dependencies among the powers of $\gamma$ where $\gamma$ is a root of $h_4(X)$ used in the proof of theorems.

\begin{lemma}
\emph{(\cite{hammons} and \cite{lint}).}
Let $\gamma$ be a root of a primitive basic irreducible polynomial $h_{4}\left(X\right)$. We have
\begin{enumerate}
\item $\pm\gamma^{j}\pm\gamma^{k}$ is invertible for $0\leq j<k<N$, $m \geq 2$. \label{enu:Lemma1.1}
\item $\gamma^{j}-\gamma^{k}\neq\pm\gamma^{l}$ for distinct $j$, $k$, $l$ in $\left[0,N-1\right]$, $m \geq 2$. \label{enu:Lemma1.2}
\item If $m\geq3$ and $i\neq j$, $k\neq l$ in $\left[0,N-1\right]$, then $\gamma^{i}-\gamma^{j}=\gamma^{k}-\gamma^{l}\Rightarrow i=k$ and $j=l$.\label{enu:Lemma1.3}
\item If $m\geq3$ and $m$ odd, then $\gamma^{i}+\gamma^{j}+\gamma^{k}+\gamma^{l}=0\Rightarrow i=j=k=l$.\label{enu:Lemma1.4}
\end{enumerate}\label{lemma:dep}
\end{lemma}

Hammons et al. \cite{hammons} determined the minimum Lee weight of $P_4$ and just showed the existence of the codewords with minimum Lee weight. Next we list all the codewords with the minimum Lee weight of $P_4$ by transforming the codewords of the extended Hamming code over $\mathbb{F}_{2^m}$.

On the following theorems, the coordinates $\{0,a,b,c,d,e\}$ correspond to the coordinates of nonzero entries of the codewords of the extended Hamming code with the Hamming weights 4 or 6.

\begin{theorem}\label{thm:list}
For odd $m \geq 3$, the minimum Lee weight of $P_{4}$ with length $2^{m}$ is $6$. The codeword $\textbf{c} = (c_\infty, c_0, \ldots, c_{N-1})$ of $P_4$ with the minimum Lee weight $6$ has one of the following forms:

\begin{table}[H]
\begin{center}
\begin{tabular}{c|c|cccccc}
Case & $c_{\infty}$ & $c_{0}$ & $c_{a}$ & $c_{b}$ & $c_{c}$ & $c_{d}$ & $c_{e}$\\
\hline Case 1 & -1 & 1 & 1 & 1 & 2 & 0 & 0\\
Case 2 & 1 & -1 & 1 & 1 & 2 & 0 & 0\\
Case 3 & 2 & -1 & 1 & 1 & 1 & 0 & 0\\
Case 4 & 0 & -1 & 1 & 1 & 1 & 2 & 0\\
\hline Case 5 & 0 & -1 & 1 & 1 & 1 & 1 & 1\\
Case 6 & 0 & -1 & -1 & -1 & 1 & 1 & 1\\
Case 7 & -1 & 1 & 1 & 1 & 1 & 1 & 0\\
Case 8 & 1 & -1 & 1 & 1 & 1 & 1 & 0\\
Case 9 & -1 & -1 & -1 & 1 & 1 & 1 & 0\\
Case 10 & 1 & -1 & -1 & -1 & 1 & 1 & 0
\end{tabular}
\end{center}\caption{List of codewords with the minimum Lee weight of $P_4$ for odd $m\geq3$}\label{table:case4odd}
\end{table}

The codewords for Cases $1$ to $4$ are obtained from the codewords of the extended Hamming code over $\mathbb{F}_{2^m}$ with the minimum Hamming weight $4$ and the codewords for Cases $5$ to $10$ are obtained from the codewords with the Hamming weight $6$.

For even $m \geq 4$, the minimum Lee weight of $P_{4}$ with length $2^{m}$ is $4$. The codeword of $P_4$ with the minimum Lee weight $4$ has one of the following $2$ forms :
\begin{table}[H]
\begin{center}
\begin{tabular}{c|c|cccc}
Case & $c_{\infty}$ & $c_{0}$ & $c_{a}$ & $c_{b}$ & $c_{c}$\\
\hline Case 11 & 1 & 1 & 1 & 1 & 0\\
Case 12 & 0 & 1 & 1 & 1 & 1
\end{tabular}
\end{center}\caption{List of codewords with the minimum Lee weight of $P_4$ for even $m\geq4$}\label{table:case4even}
\end{table}
\end{theorem}

\begin{proof}

The equation \eqref{eq:P1} implies that every codeword has even number of odd entries. Therefore Lee weight is even. So, the minimum Lee weight is even. The codeword obtained from multiplying the nonzero entries of the other codeword by $-1$ has the same Lee weight as that of the other. Therefore we may choose one of these codewords.

We assume $m$ is odd. We know the minimum Hamming weight of the extended Hamming code $H$ over a finite field is 4. If the codeword of $P_4$ has the Lee weight 6 and its image has the Hamming weight 4, then the nonzero entries are determined as $\left\{ -1,1,1,1,2\right\}.$

If the codewords of $P_4$ has the Lee weight 6 and its image has the Hamming weight 6, then we have the following 2 cases :
\begin{enumerate}
\item Codewords with nonzero entries $\left\{ -1,1,1,1,1,1\right\}. $
\item Codewords with nonzero entries $\left\{ -1,-1,-1,1,1,1\right\}. $
\end{enumerate}

We need to check that the codeword with this form satisfy equation \eqref{eq:P2}. Every element $\alpha_4 \in GR(4,m) $ has a unique representation
\begin{equation}
\alpha_4=a_{0}+2a_{1}\label{eq:rep4}\end{equation}
where $a_0, a_1 \in \mathcal{T}=\{0,1,\gamma,\ldots,\gamma^{n-1}\}$. We fix the codeword of the extended Hamming code $H$ with the minimum Hamming weight 4, so that it will satisfy $\theta^a + \theta^b + \theta^c + \theta^d = 0 $ for some $a$, $b$, $c$ and $d$. We can write this equivalently as $\theta^a \left( 1 + \theta^{b-a} + \theta^{c-a} + \theta^{d-a} \right) = 0 $. Therefore, without loss of generality, we may assume that $ 1 + \theta^a + \theta^b + \theta^c = 0 $ is satisfied for some $a$, $b$ and $c$. We choose the vector of $GR(4,m)$ with nonzero entries $\{-1,1,1,1,2\}$. We assume that $c_\infty$ is 0 or 2. The coordinates 0, $a$, $b$ and $c$ correspond to the nonzero entries of the codeword of $H$. We may assume $c_0 = -1$ without loss of generality. Otherwise, for example, for the case $c_a = -1$, it can be reduced to the case $c_0 = -1$ by multiplying $1 - \gamma^a + \gamma^b + \gamma^c = a_0 + 2a_1 $ by $\gamma^{-a}$. We put

\begin{equation}
-1+\gamma^a+\gamma^b+\gamma^c=a_{0}+2a_{1},\mbox{ } a_0,a_1\in\mathcal{T}.\label{eq:4rep3}\end{equation}

\noindent By applying the map $\tau$ in \eqref{eq:4rep3}, we can obtain
\begin{eqnarray}
a_{0}=\left(-1+\gamma^{a}+\gamma^{b}+\gamma^{c}\right)^{2^{m}} & = & 1+\gamma^{a}+\gamma^{b}+\gamma^{c}+2\left(\gamma^{2^{m-1}a}+\gamma^{2^{m-1}b}+\gamma^{2^{m-1}c}\right.\nonumber\\
 &  & \left.+\gamma^{2^{m-1}\left(a+b\right)}+\gamma^{2^{m-1}\left(a+c\right)}+\gamma^{2^{m-1}\left(b+c\right)}\right).\label{eq:4rep4} \end{eqnarray}

\noindent If we multiply \eqref{eq:4rep4} by 2, then we have \begin{equation}
2a_{0}=2\left(-1+\gamma^{a}+\gamma^{b}+\gamma^{c}\right).\end{equation}
From $2\left(-1+\gamma^{a}+\gamma^{b}+\gamma^{c}\right)=2\left(1 + \theta^a + \theta^b + \theta^c\right) = 0$, we see $a_{0}$ is zero.

We can calculate $a_{1}$ as follows: \begin{eqnarray}
2a_{1} & = & \left(-1+\gamma^{a}+\gamma^{b}+\gamma^{c}\right)-\left(-1+\gamma^{a}+\gamma^{b}+\gamma^{c}\right)^{2^{m}}\nonumber \\
& = & 2\left(1+\gamma^{2^{m-1}a}+\gamma^{2^{m-1}b}+\gamma^{2^{m-1}c}+\gamma^{2^{m-1}\left(a+b\right)}+\gamma^{2^{m-1}\left(a+c\right)}+\gamma^{2^{m-1}\left(b+c\right)}\right)\nonumber, \\
a_{1} & = & 1+\theta^{2^{m-1}a}+\theta^{2^{m-1}b}+\theta^{2^{m-1}c}+\theta^{2^{m-1}\left(a+b\right)}+\theta^{2^{m-1}\left(a+c\right)}+\theta^{2^{m-1}\left(b+c\right)}.\label{eq:4case3-2}\end{eqnarray}

\noindent By using the Frobenius automorphism, we have \[1+\theta^{2^{m-1}a}+\theta^{2^{m-1}b}+\theta^{2^{m-1}c} = \left(1+\theta^{a}+\theta^{b}+\theta^{c}\right)2^{m-1} = 0.\] Thus we have
\begin{equation}
a_{1}=\theta^{2^{m-1}\left(a+b\right)}+\theta^{2^{m-1}\left(a+c\right)}+\theta^{2^{m-1}\left(b+c\right)}.\label{eq:4case3-3}\end{equation}

\noindent Hence the coordinate $d$ such that \begin{equation}
-1+\gamma^{a}+\gamma^{b}+\gamma^{c}=2\gamma^{d}\end{equation}
is determined as $2\gamma^{d}=2\left(\theta^{2^{m-1}\left(a+b\right)}+\theta^{2^{m-1}\left(a+c\right)}+\theta^{2^{m-1}\left(b+c\right)}\right).$ Otherwise $-1+\gamma^{a}+\gamma^{b}+\gamma^{c}=0$.

For the case $-1+\gamma^{a}+\gamma^{b}+\gamma^{c}=0$, the codeword must have $c_{\infty}=2$ to satisfy the first row equation \eqref{eq:P1} and its image satisfies $\theta^{a+b}+\theta^{a+c}+\theta^{b+c} = 0$. Thus the codeword with Lee weight 6 is given as follows :
\begin{equation}
\left(\begin{array}{cccccccccc}
c_{\infty} & \ldots & c_{0} & \ldots & c_{a} & \ldots & c_{b} & \ldots & c_{c} & \ldots\\
2 & \ldots & -1 & \ldots & 1 & \ldots & 1 & \ldots & 1 & \ldots\end{array}\right).\end{equation}
This is Case 3 in Table 1 of the theorem.

For the case $-1+\gamma^{a}+\gamma^{b}+\gamma^{c}=2\gamma^{d}$ where $2\gamma^{d} = 2\left(\theta^{2^{m-1}(a+b)}+\theta^{2^{m-1}(a+c)}+\right.$ $\left.\theta^{2^{m-1}(b+c)}\right)$, the codewords must have $c_{\infty}=0$ to satisfy the first row equation \eqref{eq:P1}. Thus the codeword with Lee weight 6 is given as follows :\begin{equation}
\left(\begin{array}{cccccccccccc}
c_{\infty} & \ldots & c_{0} & \ldots & c_{a} & \ldots & c_{b} & \ldots & c_{c} & \ldots & c_{d} & \ldots\\
0 & \ldots & -1 & \ldots & 1 & \ldots & 1 & \ldots & 1 & \ldots & 2 & \ldots\end{array}\right).\end{equation}
This is Case 4 in Table 1 of the theorem. If we assume $c_\infty=\pm 1$, then we obtain Cases 1 and 2 by similar calculations and the coordinate $c$ is determined from the codeword of the extended Hamming code.

We consider the vector of $GR(4,m)$ whose image by $\eta$ is the codeword of $H$ with Hamming weight 6. Then we obtain Cases 5 and 6 if $c_\infty=0$ and Cases 7, 8, 9 and 10 if $c_\infty=\pm 1$ in the same way.

Next, we assume $m$ is even. To satisfy \eqref{eq:P1}, the codewords of $P_4$ with the Lee weight 4 have nonzero values $\{1,1,-1,-1\}$ or $\{ 1,1,1,1\}$.

If we consider the codewords with nonzero values $\{1,1,-1,-1\}$ and $c_{\infty}=0$, then it must satisfy $\gamma^a + \gamma^b - \gamma^c - \gamma^d = 0$, that is $a = d$ and $b = c$ from Lemma 2. It contradicts the assumption. On the other hand, if the codeword has $c_\infty = \pm 1$ it will also contradicts point \eqref{enu:Lemma1.2} of Lemma \ref{lemma:dep}.

Hammons et al. \cite{hammons} showed that $\gamma^{2t}+\gamma^t+1 = 0$ is satisfied where $t=(2^m-1)/3$ and the codeword with nonzero entries $c_\infty = c_0 = c_t = c_{2t} = 1$ has the Lee weight 4. The codeword with $c_\infty = 0$ and 4 other nonzero entries also has the Lee weight 4. Hence we have Case 11 and Case 12.
\end{proof}

We find the minimum Lee weight of $P_{8}$ over Galois rings of characteristic 8 by transformation of the codewords of $P_{4}$ in Tables 1 and 2.

\begin{theorem}\label{thm:lee}
The minimum Lee weight of $P_{8}$ over Galois rings of characteristic $8$ with length $2^{m}$ is $6$, for $m\geq3$ .
\end{theorem}

\begin{proof}

We know that the minimum Lee weight of $P_4$ is 6 for odd extensions and 4 for even extensions \cite{hammons}. For odd extension, if there exists the codeword with Lee weight 6 of $P_8$, then the minimum Lee weight of $P_8$ is 6 from Lemma \ref{lemma:min}.

The image by the map $\mu$ of this codeword is a codeword of $P_4$ with Lee weight 6 and the image under the map $\eta$ is a codeword of an extended Hamming code with the Hamming weight 4 or 6.

Notice that any element $\alpha_8 \in GR(8,m)$ has the unique representation \begin{equation*}
\alpha_8 = a_0 + 2a_1 + 4a_2
\end{equation*}
where $a_0,a_1,a_2 \in \mathcal{T}$ and $\mu\left(\alpha_8\right) \in GR(4,m)$ and $\tau\left(\alpha_8\right) = \alpha_8^{2^m} = a_0$.

We assume that the image of the codeword of $P_8$ by the map $\mu$ is the codeword of $P_4$ which satisfies $1+\gamma^{a}+\gamma^{b}=2\gamma^{c},a \neq b \neq c,$ for Case 1 in Table 1 of Theorem 3.

\noindent We put
\begin{equation}
1+\xi^{a}+\xi^{b}=a_{0}+2a_{1}+4a_{2},\mbox{ } a_0,a_1,a_2\in\mathcal{T}.\label{eq:rep1}\end{equation}

\noindent By applying the map $\tau$ to the equation \eqref{eq:rep1}, we can write
\begin{eqnarray}
a_{0} & = & \left(1+\xi^{a}+\xi^{b}\right){}^{2^{m}}\nonumber\\
 & =& 1+\xi^{a}+\xi^{b}+4\left(1+\xi^{2^{m-1}a}+\xi^{2^{m-1}b}\right)\cdot \left(\xi^{2^{m-2}a}+\xi^{2^{m-2}b}+\xi^{2^{m-2}\left(a+b\right)}\right)\label{eq:case1-1}\\
 &  & +6\left(\xi^{2^{m-1}a}+\xi^{2^{m-1}b}+\xi^{2^{m-1}\left(a+b\right)}\right).\nonumber \end{eqnarray}

\noindent From $4\left(1+\left(\xi^{a}\right)^{2^{m-1}}+\left(\xi^{b}\right)^{2^{m-1}}\right) = 4\left(1+\left(\theta^{a}\right)^{2^{m-1}}+\left(\theta^{b}\right)^{2^{m-1}}\right) = 4\left(1+\theta^{a}+\theta^{b}\right)^{2^{m-1}} = 0$,
we have

\begin{equation}
a_{0}=1+\xi^{a}+\xi^{b}+6\left(\xi^{2^{m-1}a}+\xi^{2^{m-1}b}+\xi^{2^{m-1}\left(a+b\right)}\right).\label{eq:case1-2}\end{equation}

\noindent We can verify $a_0 = 0$, by multiplying \eqref{eq:case1-2} by 4,
\begin{equation*}
4a_0 = 4\left(1+\xi^{a}+\xi^{b}\right)=4\left(1+\theta^{a}+\theta^{b}\right) = 0.\end{equation*}

\noindent We can calculate $a_{1}$ and $a_{2}$.
\begin{eqnarray}
2a_{1}+4a_{2} & = & \left(1+\xi^{a}+\xi^{b}\right)-\left(1+\xi^{a}+\xi^{b}\right)^{2^{m}},\nonumber \\
2a_{1}+4a_{2} & = & 2\left(\xi^{2^{m-1}a}+\xi^{2^{m-1}b}+\xi^{2^{m-1}\left(a+b\right)}\right),\nonumber \\
a_{1}+2a_{2} & = & \gamma^{2^{m-1}a}+\gamma^{2^{m-1}b}+\gamma^{2^{m-1}\left(a+b\right)}.\label{eq:case1-3}\end{eqnarray}

\noindent Let $\gamma^{u}=\gamma^{2^{m-1}a}$ and $\gamma^{v}=\gamma^{2^{m-1}b}$. We can calculate $a_{1}$ by applying the map $\tau$,
\begin{eqnarray*}
a_{1}=\left(\gamma^{u}+\gamma^{v}+\gamma^{u+v}\right)^{2^{m}} & = & \gamma^{u}+\gamma^{v}+\gamma^{u+v}+2\gamma^{2^{m-1}\left(u+v\right)}\left(1+\gamma^{2^{m-1}u}+\gamma^{2^{m-1}v}\right).\end{eqnarray*}

\noindent From $2\left(1+\left(\gamma^{u}\right)^{2^{m-1}}+\left(\gamma^{v}\right)^{2^{m-1}}\right) = 2\left(1+\left(\theta^{a}\right)^{2^{2m-2}}+\left(\theta^{b}\right)^{2^{2m-2}}\right) = 0$,
we have

\begin{eqnarray}
a_{1} = \left(\gamma^{u}+\gamma^{v}+\gamma^{u+v}\right)^{2^{m}} = \gamma^{u}+\gamma^{v}+\gamma^{u+v} = \gamma^{2^{m-1}a}+\gamma^{2^{m-1}b}+\gamma^{2^{m-1}\left(a+b\right)}.\label{eq:case1-4}\end{eqnarray}

\noindent By substituting \eqref{eq:case1-4} to \eqref{eq:case1-3}, we can
get $a_2 = 0$. Hence we have
\begin{eqnarray*}
a_{0} & = & 0, \\
a_{1} & = & \gamma^{2^{m-1}a}+\gamma^{2^{m-1}b}+\gamma^{2^{m-1}\left(a+b\right)},\\
a_{2} & = & 0.\end{eqnarray*}

\noindent Consequently, the coordinate $c$ such that
\begin{equation*}
1+\xi^{a}+\xi^{b}=2\xi^{c}\end{equation*}
is determined as $2\xi^{c}=2\left( \gamma^{2^{m-1}a}+\gamma^{2^{m-1}b}+\gamma^{2^{m-1}\left(a+b\right)}\right)$. For the cases when $c = 0,a$ and $b$, it contradicts the minimum Lee weight of $P_4$ is 6.

The codewords is as follows :

\begin{equation*}
\left(\begin{array}{cccccccccc}
c_\infty & c_0 & \ldots & c_a & \ldots & c_b & \ldots & c_c & \ldots \\
-1 & 1 & \ldots & 1 & \ldots & 1 & \ldots & 6 & \ldots
\end{array}\right).\end{equation*}
The Lee weight of this codeword is 6. In subsection 3.1, we showed the existence of a codeword of the extended Hamming code with Hamming weight 4. It ensures the existence of this codeword. We also prove the codewords corresponding to Case 9 and 10 in Table 1 has the Lee weight 6 and it ensures the existence of these codewords from the codeword of the extended Hamming code mentioned in subsection 3.1.

Next we assume the image of the codeword
of $P_8$ by the map $\mu$ satisfies $1+\gamma^{t}+\gamma^{2t}=0$ where $\gamma^{3t}=1$
and $t=\frac{2^{m}-1}{3}$ for Case 11 of Table 1 of Theorem 3. So, $1+\theta^t+\theta^{2t} = 0$ is also satisfied.

\noindent We put\begin{equation}
1+\xi^{t}+\xi^{2t}=a_{0}+2a_{1}+4a_{2},\mbox{ } a_0,a_1,a_2\in\mathcal{T}.\label{eq:rep11}\end{equation}

\noindent By applying the map $\tau$ in \eqref{eq:rep11}, we can write
\begin{eqnarray*}
a_{0} & = & \left(1+\xi^{t}+\xi^{2t}\right)^{2^{m}}\nonumber\\
 & = & 1+\xi^{t}+\xi^{2t}+4\left(1+\xi^{2^{m-1}t}+\xi^{2^{m-1}2t}\right)\cdot \left(\xi^{2^{m-2}t}+\xi^{2^{m-2}2t}+1\right)\\
 &  & +6\left(\xi^{2^{m-1}t}+\xi^{2^{m-1}2t}+1\right). \end{eqnarray*}

\noindent From $4\left(1+\left(\xi^{t}\right)^{2^{m-1}}+\left(\xi^{2t}\right)^{2^{m-1}}\right) = 4\left(1+\left(\theta^{t}\right)^{2^{m-1}}+\left(\theta^{2t}\right)^{2^{m-1}}\right) = 4\left(1+\theta^{t}+\theta^{2t}\right)^{2^{m-1}} = 0$, we obtain\begin{equation*}
a_{0}=1+\xi^{t}+\xi^{2t}+6\left(\xi^{2^{m-1}t}+\xi^{2^{m-1}2t}+1\right).\end{equation*}
Thus we verify $a_0 = 0$ from $4a_0 = 0$. We can calculate $a_{1}$ and $a_{2}$.
\begin{eqnarray}
2a_{1}+4a_{2}& = & \left(1+\xi^{t}+\xi^{2t}\right)-\left(1+\xi^{t}+\xi^{2t}\right)^{2^{m}},\nonumber \\
2a_{1}+4a_{2}& = & 2\left(\xi^{2^{m-1}t}+\xi^{2^{m-1}2t}+1\right),\nonumber \\
a_{1}+2a_{2}& = & \gamma^{2^{m-1}t}+\gamma^{2^{m-1}2t}+1.\label{eq:case11-2}\end{eqnarray}

\noindent From applying the Frobenius automorphism, we have $\left(1+\gamma^t+\gamma^{2t}\right)^{2^{m-1}} = 0$, that is $a_1 + 2a_2 = 0$. Thus we obtain $a_0 = a_1 = a_2 = 0$, then \begin{equation*}
1+\xi^{t}+\xi^{2t}=0.\end{equation*}

\noindent The codewords is as follows :\begin{equation*}
\left(\begin{array}{ccccccc}
c_\infty & c_0 & \ldots & c_t & \ldots & c_{2t} & \ldots\\
5 & 1 & \ldots & 1 & \ldots & 1 & \ldots \end{array}\right)\end{equation*}
and the Lee weight of this codeword is 6.
It concludes that the minimum Lee weight of $P_{8}$ is 6.
\end{proof}

We will give some examples of the codewords with the minimum Lee weight 6 to clarify our proof.

\vspace{0.5cm}
\textbf{Examples of the codewords with the minimum Lee weight of $P_8$}

\begin{enumerate}
\item We give the example of the codeword with minimum Lee weight which generated by searching the weight of the codewords over $GR(8,m)$ by computer programs.
\begin{itemize}
\item For case $m=3$, we found the codeword $\textbf{c} = \left(7,1,1,6,0,0,1,0\right)$ has the minimum Lee weight 6 by computer search.
\item For case $m=4$, we found the codeword $\textbf{c} = \left(7,1,6,0,0,0,0,0,0,0,0,0,1,1,0,0\right)$ and $\textbf{c} = \left(5,1,0,0,0,0,1,0,0,0,0,1,0,0,0,0\right)$ have the minimum Lee weight 6 by computer search. By consider Table \ref{table:case4odd} and Table \ref{table:case4even} of Theorem \ref{thm:list}, we can see that the first codeword transformed from Case 1 and the second codeword transformed from Case 11.
\end{itemize}

\item For $m=3$, it is easy to see $1 + \theta + \theta^5 = 0$ is satisfied where $\theta$ is a root of the primitive polynomial $h_2(X) = X^3 + X^2 + 1$. Next we assign the value of $c$ which satisfies $1 + \gamma + \gamma^5 = 2\gamma^c$, where $\gamma$ is a root of the primitive basic irreducible polynomial $h_4(X) = X^3 - X^2 - 2X - 1$. We obtain $c=2$. From the equation $1 + \gamma + \gamma^5 = 2\gamma^2$, it is sufficient to determine the sign of the equation $1 + \xi + \xi^5 = \pm 2\xi^2$, where $\xi$ is a root of the primitive basic irreducible polynomial $h_8(X) = X^3 - 5X^2 -6X - 1$. Thus we obtain the codeword $\textbf{c} = \left(7,1,1,6,0,0,1,0\right)$ of $P_8$ with the minimum Lee weight 6. We do the same scheme for the case $m=4$ and know the codeword $\textbf{c} = \left(7,1,1,0,0,1,6,0,0,0,0,0,0,0,0,0\right)$ has the minimum Lee weight 6.

\end{enumerate}

%%% References
%% Note: use of BibTeX also works

\bibliographystyle{plain}

\end{document}